\documentclass[11pt, onecolumn]{IEEEtran} 
\usepackage{geometry}  
\usepackage{hyperref}
\hypersetup{backref, colorlinks=true, linkcolor=blue, citecolor=blue, urlcolor = blue}
\usepackage{graphicx,wrapfig,lipsum}	
\usepackage{yfonts}	

\usepackage{amssymb}
\usepackage{amsmath}
\usepackage{amsthm}
\usepackage{scalerel}
\usepackage{verbatim}

\makeatletter
\newtheorem*{rep@theorem}{\rep@title}
\newcommand{\newreptheorem}[2]{%
	\newenvironment{rep#1}[1]{%
		\def\rep@title{#2~\ref{##1}}%
		\begin{rep@theorem}}%
		{\end{rep@theorem}}}
\makeatother

\newtheorem{thm}{Theorem}[section]

\newreptheorem{thm}{Theorem}

\newtheorem{lem}[thm]{Lemma}

\newtheorem{prop}[thm]{Proposition}
\newreptheorem{prop}{Proposition}

\newreptheorem{cor}{Corollary}

\newreptheorem{conj}{Conjecture}
\newtheorem{rem}[thm]{Remark}

\theoremstyle{definition}
\newtheorem{defn}{Definition}

\usepackage[dvipsnames]{xcolor}

\def\bigE{{\mathbb{E}}}

\allowdisplaybreaks[1]

\title{Stabilizing a system with an unbounded random gain using only a finite number of bits}
\author{Victoria Kostina\IEEEauthorrefmark{1}, Yuval Peres\IEEEauthorrefmark{2}, Gireeja Ranade\IEEEauthorrefmark{2}, Mark Sellke\IEEEauthorrefmark{3}\\
	\IEEEauthorblockA{\IEEEauthorrefmark{1}California Institute of Technology, Pasadena, CA }\\
	\IEEEauthorblockA{\IEEEauthorrefmark{2} Microsoft Research, Redmond, WA}\\
	\IEEEauthorblockA{\IEEEauthorrefmark{3} University of Cambridge, Cambridge, UK}\\
	vkostina@caltech.edu, peres@microsoft.com, giranade@microsoft.com, msellke@gmail.com
}

\date{}							
\begin{document}
\maketitle

\section*{Abstract}
 We study the stabilization of an unpredictable linear control system where the controller must act based on a rate-limited observation of the state. More precisely, we consider the system $X_{n+1}=A_nX_n+W_n-U_n$, where the $A_n$'s are drawn independently at random at each time $n$ from a known distribution with unbounded support, and where the controller receives at most $R$ bits about the system state at each time from an encoder. We provide a time-varying achievable strategy to stabilize the system in a second-moment sense with fixed, finite~$R$. 

While our previous result provided a strategy to stabilize this system using a variable-rate code, this work provides an achievable strategy using a fixed-rate code. The strategy we employ to achieve this is time-varying and takes different actions depending on the value of the state. It proceeds in two modes: a normal mode (or zoom-in), where the realization of $A_n$ is typical, and an emergency mode (or zoom-out), where the realization of $A_n$ is exceptionally large.  
\section{Introduction}
System design for decentralized control over communication networks requires an understanding of the informational bottlenecks that affect our ability to stabilize the system. This paper focuses on a system that grows unpredictably and is observed over a rate-limited channel. We consider a modification of the classical data-rate theorems~\cite{wong1997systems,tatikonda,nair2000stabilization,nair2007feedback}, where the system growth is unpredictable and random, and provide a strategy to control such a system over finite and fixed-rate channel. This builds on previous work~\cite{kostina_rate-limited_2016}, which considered the control of a system with unpredictable growth with unbounded support over a channel using a variable-rate code.

Specifically, we consider the control of the following system:
\begin{equation} \label{eq:system}
X_{n+1} = A_{n} X_{n} + W_{n} - U_{n}.
\end{equation} 
Here $X_{n}$ is the (scalar) state of the system at time $n$. The system gains $\{A_{n}\}_{n \geq 0}$ are drawn i.i.d.\ from a known distribution and model the uncertainty the controller has about the system. The additive disturbances $\{W_{n}\}_{n \geq 0}$ are also i.i.d. drawn from a known distribution. The controller chooses the control $U_{n}$ causally based on $R$-bit observations $Y_{0}^{n} := \left\{ Y_{0}, \dots, Y_{n} \right\}$, where $Y_n \in \{1, 2, \ldots 2^R\}$, which are transmitted by an encoder co-located with the system. Our goal is to stabilize the system in a second-moment sense, which thus requires the co-design of an encoder-controller pair. 

\begin{figure}[tbp]
	\begin{center}
		\includegraphics[width=.4\textwidth]{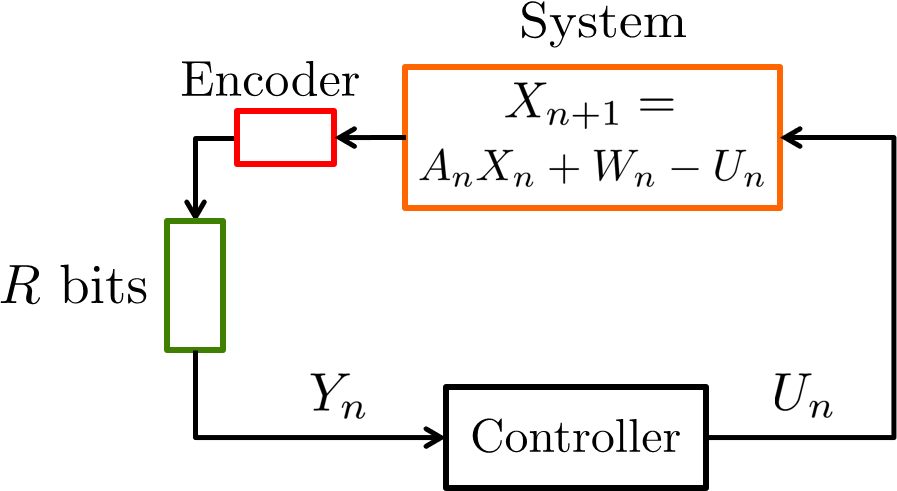}
		\caption{Rate-limited control: The control $U_{n}$ can be any function of the observations $Y_{0}^{n}$. The encoder chooses exactly $R$ bits to transmit to the controller at every time-step through $Y_n$, i.e. $Y_n \in \{1, 2, \ldots 2^R\}$. The system gains $\{A_{n}\}_{n \geq 0}$ are i.i.d.\ random variables, and so are the additive disturbances $\{W_{n}\}_{n \geq 0}$.}
		\label{default}
	\end{center}
\end{figure}

It is known that without a rate-limit on the channel, the system in~\eqref{eq:system} is second-moment stabilizable if and only if $\sigma_{A}^{2} < 1$, where $\sigma_{A}^{2}$ is the variance of $A_{n}$~\cite{uncertaintyThreshold}. Hence, we focus on the case \mbox{$\sigma_A^2 < 1$}. 
The goal is to guarantee second-moment stability of the system. Our main theorem is stated below.

\begin{thm}\label{thm:main}
	Assume that $A_{n}$ and $W_{n}$ have finite $(4 + \epsilon)$-th moments. Then, for some $R \in \mathbb N$, there exists a rate-$R$-limited control strategy (see Fig.~\ref{default}) that achieves second-moment stability, i.e.,
\begin{equation}
 \limsup_{n \to \infty} \bigE[ X_{n}^{2} ] < \infty. \label{2momentstability}
\end{equation}
\end{thm}

\subsection{Strategy overview}
Previous works (e.g.~\cite{martinsUncertain,okano2014arxiv, kostina_rate-limited_2016}) 
have investigated variable-rate coding strategies for stabilizing the second-moment of the system~\eqref{eq:system}, as well as other related systems. 
However, variable-rate strategies fundamentally require that the encoder send an unbounded number of bits in the case when the realization of $A_{n}$ is large, and cannot be directly adapted for use in a fixed-rate regime.


In contrast to the variable-rate strategies, we employ a time-varying strategy in this work. It works in two different ``modes,'' depending on whether the realization of $A_n$ was typical or exceptionally large. We will describe the strategy precisely later, but provide an intuitive description here. 

At every time-step, the controller maintains an interval that estimates system state. When the state remains in the interval that was estimated by the controller, the strategy proceeds in ``normal'' (zoom-in) mode.  If the state escapes the predicted interval, usually due to the realization of an exceptionally large $A_n$, the strategy enters ``emergency'' (zoom-out) mode. There, the controller grows it's guess for the interval at every time-step till we are able to locate $X_n$, following which we step back into normal mode. This time-varying strategy fundamentally works because we can guarantee some average decrease during normal mode, while we can bound the blow-ups from emergency mode (since we assume that the $4$th moments of $A$ and $W$ are bounded). 

\subsection{Related work}\label{sec:related}
Fundamentally, our problem setup grows out of the uncertainty threshold principle~\cite{uncertaintyThreshold} and the extensive work to understand data-rate theorems~\cite{wong1997systems,tatikonda,nair2000stabilization,nair2007feedback,victoriacontrol}. However, the work on the uncertainty threshold principle~\cite{uncertaintyThreshold} does not consider rate limits on the observations, and the classical data-rate theorems assume that the system model and growth rate are always perfectly known to the controller. 

Martins et al.~\cite{martinsUncertain} were the first to consider the  rate-limited control of a system with uncertain growth. Their achievable scheme is quantization-based and assumes that the uncertainty on the system growth is bounded above. 
Phat~et~al.~\cite{phat2004robust} also consider rate-limited control with uncertain parameters from a robust control perspective. Their setup differs from ours in that it is not stochastic, and only considers bounded support for the uncertainty on the parameters. They provide a uniform quantization scheme that can stabilize their system.

Okano and Ishii have made significant progress on understanding rate-limited control of systems with unpredictable growth from a worst-case perspective~\cite{okano2012data,okano2012datalossy,okano2014arxiv}. 
However, they also consider systems where the uncertain parameters have bounded support and do not consider additive noise in their model. The achievable scheme in~\cite{okano2014arxiv} proposes a non-uniform optimal quantizer for their problem that uses bins with logarithmically decreasing lengths, with the bins closest to zero being the largest. However, this 
cannot work in the setting where both $A_{n}$ and $X_{n}$ can have unbounded support, as is the case in our work.



We also build on a series of investigations around model unpredictability and multiplicative noise in control systems~\cite{controlcapacity,gireejaAllerton,tiger}. In this work, we take a stochastic control approach, which complements the investigations around parameter uncertainty in the robust control literature~\cite{fu2010quantized,zhou1998essentials}.

In the context of linear systems with known parameters and unbounded additive disturbances, the necessity of adaptive quantization for stabilization has been long recognized. Nair and Evans \cite{nairStabilization} proved that time-invariant fixed-rate quantizers are unable to attain bounded cost if the noise is unbounded \cite{nairStabilization}, regardless of their rate. The reason is that since the noise is unbounded, over time, a large magnitude noise realization will inevitably be encountered, and the dynamic range of the quantizer will be exceeded by a large margin, not permitting recovery. Adaptive quantizers of zooming type were originally proposed by Brockett and Liberzon \cite{brockett2000quantized}, and further studied in \cite{nairStabilization,yuksel2010fixedrate,kostina2018exact}.   Such quantizers ``zoom out'' (i.e. expand their quantization intervals; this corresponds to our ``emergency mode'') when the system is far from the target and  ``zoom in'' when the system is close to the target (this corresponds to our ``normal mode''). This paper extends our recent analysis \cite{kostina2018exact} of adaptive quantization for known linear systems to unpredictably varying ones.

\section{Setup}

We consider the system in~\eqref{eq:system}. 
Here $W_n, A_n$ are i.i.d. variable noises drawn from the laws $W$ and $A$. These are independent of the system state $X_{n}$ and the control term $U_n$. We use $n$ as the time index. Without loss of generality, we assume our initial condition $X_0=0$. The \emph{encoder} and \emph{controller} will work together to stabilize the sequence $(X_n)$ in the second-moment sense, i.e., \eqref{2momentstability}.
The encoder observes $X_n$ perfectly, and transmits a symbol from the set $\{1, 2, \ldots, 2^{R}\}$ to the controller at each time-step. $R$ is a finite universal constant that we will choose later. The controller can choose $U_n$ as a function of these $R$ bits and any previous history. Let $\mu_A,\sigma_A,\mu_W,\sigma_W$ be the means and standard deviations of $A,W$. Without loss of generality, set $\mu_W=0$.

The goal of second-moment stability cannot be accomplished for arbitrary distributions $A,W$. 
Using the independence of $A_{n}$, $W_{n}$ and $U_{n}$, we see from~\eqref{eq:system} that~\cite{uncertaintyThreshold}: 
\begin{equation}
\mathbb E[X_{n+1}^2]=\sigma_A^2\mathbb E[X_n^2]+\mathbb E[(\mu_AX_n-U_n)^2]+\sigma_W^2. \label{eq:basic}
\end{equation}
It is clear from \eqref{eq:basic} that the assumption $\sigma_A<1$ is required for second-moment stabilizability (except possibly in the case $\sigma_A=1,\sigma_W=0$ which we ignore.) Our result will provide a strategy to achieve it. We require the mild assumption that $A$ and $W$ have finite $4+\varepsilon$-moment for some $\varepsilon>0$.

\section{Description of the Strategy}

At the $n$th time-step, three things happen in the following order:

\begin{enumerate}

\item The encoder transmits a codeword from a codebook of size $R$ bits to the controller based on the observation of $X_n$. The exact value of $R$ is described below.

\item The controller chooses $U_n$ based on the transmission from the encoder and any past history.

\item The noises $A_n, W_n$ are generated and the new state $X_{n+1}$ is determined, based on $A_n, W_n$ and $U_n$. This value of $X_{n+1}$ is observed by the encoder.

\end{enumerate}

We proceed in ``rounds,'' which are blocks of time-steps. Most rounds will consist of a single time-step, in those cases where $X_{n}$ falls in an interval predicted by the encoder-controller pair. In those cases where $X_{n}$ falls outside those bounds, a round might consist of multiple time-steps, until an appropriate bounding interval is found for the state.


Throughout, we will maintain positive numbers $I_n\leq M_n,$ that will be defined below. Roughly speaking, $M_n$ will represent an estimate of the maximum value of $|X_n|$, and $I_n$ will capture the quantization error in the controller's estimate of $X_n$.

%

We fix $M_0$ (a lower bound for the $M_n$'s), a large constant $P$, and a small $\delta$ to be chosen later, such that $\frac{1}{\delta} \in \mathbb{Z}$. The constant $P$ captures the controller's guess on how the $M_n$'s will grow at each time-step. The constant $\delta$ specifies the size of the codebook to $\frac{2}{\delta} + 1$, and the transmission rate is $R = \lceil \log (\frac{2}{\delta} + 1) \rceil.$

We will now describe how the $M_n$'s will evolve with time along with our strategy. We first initialize $M_1=I_1=M_0$. 
At each time-step, we check how $X_n$ compares to the guess of the controller. If $|X_n| \leq P M_{n-1}$, we operate in normal mode. In this case, the encoder transmits a quantized version of the state to the controller. If $|X_n| > P M_{n-1}$ we operate in emergency mode, and the encoder simply sends a special codeword to indicate this emergency mode. A \emph{round} of time-steps will continue until we exit emergency mode at a time $\ell$ such that $X_\ell \leq PM_{\ell-1}$.\\

\noindent \textbf{Case 1:} $X_n \leq P M_{n-1}$.

In this case, $X_n$ has landed within the bounds we expected it to land in, and we are starting the round in \emph{normal mode}. The encoder partitions the interval $[-P M_{n-1}, P M_{n-1}]$ into $\frac{2}{\delta}$ subintervals of length $\delta P M_{n-1}$ and sends a codeword to represent the index of the interval $(a_{n},b_{n})$ containing $X_{n}$. Set
\begin{align}
M_{n}&=\max(M_0,|a_{n}|,|b_{n}|),\\
I_{n}&=\max\left(M_0,\frac{|a_{n}-b_{n}|}{2}\right),\label{eq:idef}\\
\rho_{n+1}&= \textrm{sgn}(a_{n}) = \textrm{sgn}(b_{n}), 
\label{eq:rhosign}
\end{align}
where equality in \eqref{eq:rhosign} holds because $\frac{1}{\delta}$ is an integer, and so the point $0$ is always an endpoint of a quantization interval.
The values $M_n,I_n,\rho_n$ are common knowledge between the encoder and controller; they are uniquely determined by the bits that have been sent by the encoder. Consequently, in normal mode, the controller knows which interval of the form 
\begin{equation}\label{eq:interval}
\rho_n[M_n-2I_n,M_n]
\end{equation}
contains $X_n$, where $\rho_n \in \{\pm 1\}$.

Now the controller will set 
\begin{equation}\label{eq:u}
U_n=\rho_n\mu_A(M_n-I_n).
\end{equation}
The random variables $A_n,W_n$ are generated, and $X_{n+1}$ is determined. 
Hence, we will have that:
\begin{equation} \label{eq:I}
|\mu_A X_n-U_n|\leq \mu_A I_n.
\end{equation}

If $|X_{n+1}|\leq PM_n$, the \emph{round} ends at time-step $n+1$. Else, it will proceed for at least one more time-step.\\

\noindent \textbf{Case 2:} $|X_{n}|>PM_{n-1}$.
    
If this is the case, the encoder simply uses the special ($\frac{2}{\delta} + 1$-th) codeword to indicate that we are in \emph{emergency mode} and the state has escaped the predicted interval. The \emph{round} will continue until we exit emergency mode.

Throughout emergency mode, we will take \[M_{n+k+1}=PM_{n+k}\] until exiting. More specifically, if at the start of time-step $n+k$ we are in emergency mode, the following steps happen:

\begin{enumerate}

\item The encoder uses the special codeword to indicate that we are in emergency mode.

\item The controller sets $U_{n+k}=0$. Since the controller does not have a good estimate of the state, it does nothing.

\item $A_{n+k},W_{n+k}$ are generated, and $X_{n+k+1}$ is determined. If $X_{n+k+1} < P M_{n+k}$, we exit emergency mode and the round ends. Else, we continue to the next time-step in emergency mode.


\end{enumerate}
\begin{rem}
Note that the state always belongs to the interval indicated by~\eqref{eq:interval} at the end of a round.
%
We also see that $M_n,I_n\geq M_0$ at all times.
\end{rem}

\section{Preparations for Analysis}

\subsection{A Modified Sequence}


The idea is to show that $|X_n|^2$ decreases on average when it is large. However this is not true during emergency mode, so we will define a dominating sequence $(N_{n})$ which does have the desired property. The analysis will also use a trick of modifying the sequence $(X_n)$ slightly. We will freeze the value of $X_n$ at a time of interest $n_0$, by setting the subsequent $A_n$'s to $1$ and $W_n$'s to $0$. This modification is an important step required to obtain the dominating sequence though Proposition~\ref{prop:trivial} as explained in Remark~\ref{rem:prop}. If we continued to let the state $X_n$ evolve, then we would have to account for the probability that $X_{n_0}$ was large but $A_{n_0}$ was small.

First, we make a few important definitions. Fix $n_0$, where $\mathbb E[X_{n_0}^2]$ is what we would like to bound. We define a modified sequence as follows: 
\begin{equation}
 \tilde X_n=\begin{cases} X_n &\mbox{if } n<n_0, \\
 X_{n_0} & \mbox{if } n \geq n_0. \end{cases}
\end{equation}
which satisfies the recurrence 
\begin{equation}
 \tilde X_{n+1}=\tilde A_n\tilde X_n+\tilde W_n-\tilde U_n
\end{equation}
for 
\begin{equation}
(\tilde A_n,\tilde W_n,\tilde U_n)=\begin{cases} (A_n,W_n,U_n) &\mbox{if } n<n_0 \\
(1,0,0) & \mbox{if } n \geq n_0 \end{cases}
\end{equation} 
Define a sequence $\tilde M_n$:
%
\begin{equation}
 \tilde M_{n}=\begin{cases} M_n &\mbox{if } n\leq n_0 \\ P\tilde M_{n-1} &\mbox{if } n>n_0 \text{ and } |\tilde X_{n_0}|>\tilde M_{n-1}\\ \tilde M_{n-1}&\mbox{if } n>n_0 \text{ and } |\tilde X_{n_0}|\leq \tilde M_{n-1}.\end{cases}
\end{equation}
Note that $\tilde M_n$ might still be growing for $n > n_0$, even though the sequence $\tilde X_n$ it aims to bound is not growing for $n > n_0$.
Finally define:
\begin{equation}
 \tilde I_{n}=\begin{cases} I_n &\mbox{if } n\leq n_0\\  I_{n_0}&\mbox{if } n> n_0.\end{cases}
\end{equation}

Now we define a counter for the length of a round for the modified sequence $\tilde{X}_{n}$. 

\begin{defn}

Define $\tau(n)$ by 
\begin{equation}
 \tau(n)=\min\left\{m\geq n:|\tilde X_m|\leq P \tilde M_{m-1}\right\}.
\end{equation}
Assuming time-step $n<n_0$ does not involve an emergency which lasts until time $n_0$, the value $\tau(n)$ is simply the first time-step $n$ or larger which starts in normal mode. Hence for a time-step $n$ which exits in normal mode, $\tau(n)=n$. For a time-step $n$ which starts in emergency mode, suppose that time-steps $n,...,n+k<n_0$ exit in emergency mode and $n+k+1$ does not. Then 
\begin{equation}
 \tau(n)=\tau(n+1)=...=\tau(n+k)=n+k+1.
\end{equation}

\end{defn}
Note that we always have $\tau(\tau(n))=\tau(n)\geq n$, usually with equality. 


\begin{defn}
Let $K$ be a large constant.
Define 
\begin{equation} \label{eq:Q}
 Q_n=\sqrt{\tilde M_n^2+K\tilde I_n^2}.
\end{equation}

\end{defn}

The strange-seeming $Q_n$ is essentially $M_n$ but with better expected-decrease guarantees that will be relevant later.

Finally, we define the dominating sequence which we will show is decreasing on average for large~$n$. With this definition of $N_{n}$ we ``front-load'' the cost of an emergency mode. Thus, $N_{n}$ is decaying exponentially during the emergency mode, even though $X_{n}$ is not, which helps us obtain the required bound.
\begin{defn}

For a time-step $n$, let 
\begin{equation}
 N_n=Q_{\tau(n)}2^{\tau(n)-n}. \label{eq:Nn}
\end{equation}

\end{defn}
This sequence essentially measures the time until the end of the round from any $n$ through the exponent $\tau(n)-n$. The idea is to show that $N_n$ decreases on average when it is large, and to conclude from this that $\mathbb E[N_n^2]$ is uniformly bounded. We show in Proposition~\ref{prop:trivial} below that $N_{n_0}$ bounds $X_{n_0}$, allowing us to conclude that $\mathbb E[X_{n_0}^2]$ is bounded, for all $n_0$.

For technical reasons explained just below, we will work with the sequence $\tilde X_n$ when estimating $\mathbb E[X_{n_0}^2]$. We think of this modified sequence as evolving normally until time $n_0$. At time $n_0$, we freeze the value of $X_n$ to be $X_{n_0}$. We finish the current round if we are in the middle of emergency mode, and then stop permanently. 

\begin{prop}

\label{prop:trivial}

We always have 
\begin{equation}
 |X_{n_0}|\leq N_{n_0}.
\end{equation}

\end{prop}

\begin{proof}

Since $X_{n_0} = \tilde X_{n_0}$ and the value of $\tilde X_n$ is frozen to be $\tilde X_{n_0}$ for all  $n \geq n_0$ and $\tau(n_0)\geq n_0$, we have that
\begin{equation}
 |\tilde X_{n_0}|=|\tilde X_{\tau(n_0)}|\leq M_{\tau(n_0)}\leq Q_{\tau(n_0)}\leq N_{n_0}.
\end{equation}

\end{proof}

\begin{rem}\label{rem:prop}

If we did not switch to working with the sequence $(\tilde X_n)$, then the inequalities above would hold with $X$ in place of $\tilde X$ but the first equality would not hold. The trick of freezing the value of $\tilde{X}_{n}$ at $X_{n_{0}}$ for $n\geq n_{0}$ allows us to focus exactly on bounding the quantity of interest, which is $\bigE[X_{n_{0}}^{2}]$, without worrying about fluctuations in the sequence due to $A_{n}$'s for $n \geq n_{0}$.


\end{rem}

\section{Analysis of the Algorithm}

Now we state our main theorem showing that $N_n$ decreases on average when large. We need to specify a filtration; we will include in the $\sigma$-algebra $\mathcal F_n$ all steps of the algorithm which have happened so far. 

\begin{defn}

Let $\mathcal F_n$ be the $\sigma$-algebra generated by $(A_k)_{k<n}$, $(W_k)_{k<n}$ and all bits sent by the encoder in time-steps up through $n-1$. 
\label{eq:Fn}
\end{defn}

%
%

\begin{thm}

\label{thm:mainestimate}
Let $n_0$ be any fixed time. Then, for any $n<n_0$, positive constant $c<\min\left(1-\sigma_A^2,\frac{3}{4}\right)$ that does not depend on $n_{0}$, sufficiently large $P,M_0,K$, and sufficiently small $\delta=\delta(c,P,M_0,K)$, there exists a constant $C(c, P, M_0, K, \sigma_W,\sigma_A)$ such that 
\begin{equation}
\mathbb E[N_{n}^2]\leq C
\end{equation}
%

\end{thm}

Our main result, Theorem~\ref{thm:main} follows immediately as a corollary.

%

\begin{proof}[Proof of Theorem~\ref{thm:main}]
        For any fixed $n_0$, modifying the sequence at $n_0$ results in a sequence $N_n$. We know from Prop~\ref{prop:trivial} that $|X_{n_0}| \leq N_{n_0}$. 
        %
        Theorem~\ref{thm:mainestimate} implies inductively that 
\begin{equation}
 \sup_{n}(\mathbb E[N_n^2]) = \sup_{n_0}\left(\sup_{n<n_0} \mathbb E[N_n^2]\right) <\infty,
\end{equation}
        %
        which gives the result.
\end{proof}

\subsection{Proof of Theorem~\ref{thm:mainestimate}}

For a time-step $n<n_{0}$ that is in normal mode, we have that $\tau(n)=n$. We write:
\begin{equation}
 N_{n+1}^2= N_{n+1}^2 \mathbf{1}_{\tau(n+1)=n+1}+N_{n+1}^2 \mathbf{1}_{\tau(n+1)>n+1}. \label{eq:split}
\end{equation}
We estimate the conditional expectation of each term separately. For the first term we establish a second-moment averaged decrease from $M_n$ to $M_{n+1}$ when we stay in normal mode, while for the second we show that the contributions from emergency mode are small. The first step is elementary but with slightly involved algebra, while the second requires us to control the tails for the length of a round (and thus requires a bounded $4+\epsilon$-th moment). We split this work into the following pair of lemmas.

\begin{lem}
\label{lem:normalbound}
Suppose that time-steps $n$ and $n+1$ are both $<n_0$ and start in normal mode, so that $\tau(n) =n$ and $\tau(n+1) =n+1$. 
Then for any positive constant $c<1-\sigma_A^2$, sufficiently large $K$, arbitrary $P,M_0$, and sufficiently small $\delta=\delta(c,P,M_0,K)$ we have
\begin{equation} \label{eq:lemma1}
 \mathbb E[N_{n+1}^2\mathbf{1}_{\tau(n+1)=n+1}\mid \mathcal{F}_n]\leq (1-c)N_n^2+ 2\sigma_W^2+K M_0^2~\text{a.s.}.
\end{equation}

\end{lem}

\begin{lem}

\label{lem:emergencybound}
Suppose that $A$ and $W$ have finite $4 + \epsilon$-th moments. Further, suppose that time-step $n$ starts in normal mode, so that $\tau(n) = n$. However, \mbox{$X_{n+1} > P M_n$}, so that $\tau(n+1)> n+1$. Then, 
\begin{equation}
 \mathbb E[N_{n+1}^2\mathbf{1}_{\tau(n+1)>n+1} \mid \mathcal{F}_n]\leq \varepsilon N_n^2~\text{a.s.},
\end{equation}
where $\varepsilon(P,M_0)$ may be made made arbitrarily small by choosing $M_0$ sufficiently large and then $P$ sufficiently large.

\end{lem}

We now briefly explain how to derive Theorem~\ref{thm:mainestimate} from Lemmas~\ref{lem:normalbound} and~\ref{lem:emergencybound}.
\begin{proof}[Proof of Theorem~\ref{thm:mainestimate}, using Lemmas~\ref{lem:normalbound} and~\ref{lem:emergencybound}] 

First, we consider the case when we are in normal mode and $\tau(n) = n$, and use the representation in~\eqref{eq:split}. For $c$ as in the theorem statement, we first pick $P,M_0$ large enough such that, in the language of Lemma~\ref{lem:emergencybound}, $\varepsilon(P,M_0)$ satisfies 
\begin{equation}
 c+\varepsilon<\min\left(1-\sigma_A^2,\frac{3}{4}\right).
\end{equation}  
        \noindent Lemma~\ref{lem:emergencybound} then gives: 
\begin{equation}
 \mathbb E[N_{n+1}^2\mathbf{1}_{\tau(n+1)>n+1} \mid \mathcal{F}_{\tau(n)}]\leq \varepsilon N_n^2. \label{eq:step1}
\end{equation}

        \noindent Then substituting $c+\varepsilon$ for $c$ in Lemma~\ref{lem:normalbound}, pick $K$ large and then $\delta$ small to make 
\begin{equation}
 \mathbb E[N_{n+1}^2 \mathbf{1}_{\tau(n+1)=n+1} \mid \mathcal{F}_{\tau(n)}]\leq (1-c-\varepsilon)N_n^2+2\sigma_W^2+K M_0^2 \label{eq:step2}
\end{equation}
        \noindent hold whenever $\tau(n)=n$. Combining~\eqref{eq:step1} and~\eqref{eq:step2} using~\eqref{eq:split}, we have 
        \begin{equation}
        \mathbb E[N_{n+1}^2 \mid \mathcal{F}_{\tau(n)}]\leq (1-c)N_n^2+ 2\sigma_W^2+K M_0^2. \label{eq:indstep}
        \end{equation}

        \noindent in this case. 

Now consider the case that we are in emergency mode and $\tau(n)>n$. Here, due to \eqref{eq:Nn} we automatically have $N_{n+1}=\frac{N_n}{2}$. Since $c<\frac{3}{4}$ and $\mathcal{F}_{\tau(n)}$ includes $N_n$, it follows that: 
\begin{equation}
  \mathbb E[N_{n+1}^2 \mid \mathcal{F}_{\tau(n)}]= \frac{1}{4} N_n^2 \leq (1-c) N_n^2.
\end{equation}
        Thus~\eqref{eq:indstep} also holds when $\tau(n)>n$. 

\noindent Let 
\begin{equation}
D = 2\sigma_W^2 + (1+K) M_0^2. \label{eq:D}
\end{equation}
and observe that \eqref{eq:indstep} implies
\begin{equation}
 \mathbb E[N_{n+1}^2 ]\leq (1-c) \mathbb E[N_n^2]+ 2\sigma_W^2+K M_0^2. \label{eq:indstepE}
\end{equation}

\noindent Now we can proceed using induction. Assume $\mathbb E[N_{n}^2] \leq \frac{D}{c} = C$. Then, we have from \eqref{eq:indstepE} that
\begin{align}
        \mathbb E[N_{n+1}^2]
        &\leq (1-c) \frac{D}{c} + D - M_0^2 \\
        &< \frac{D}{c}.
\end{align} 
        It remains to verify the base case. Since \mbox{$X_0 = 0$}, we know that $M_0 = I_0$, the initial \mbox{$\tau(0) = 0$} and \mbox{$Q_0 = (1+K)M_0^2$}. Hence $N_{0} = Q_{0} \cdot 2^{0} = Q_{0}$
and \mbox{$\mathbb E[N_{0}^2] \leq D \leq \frac{D}{c},$ since $c<1$.}

        This proves Theorem~\ref{thm:mainestimate} assuming the lemmas.
\end{proof}


\subsection{Proofs of Lemmas}
Finally, we prove the two key lemmas below.

\begin{proof}[Proof of Lemma~\ref{lem:normalbound}]

In the case considered here, we have $\tau(n) = n$, and this implies $\mathcal{F}_{\tau(n)} = \mathcal{F}_{n}$ and $N_{n} = Q_{n}$.  Further, we have $\tau(n+1) = n+1$, and so $N_{n+1}=Q_{n+1}$. Hence, for the remainder of this proof we will refer to the $Q$'s instead of the $N$'s. Similarly, since $n < n_{0}$ we will refer to the $X$'s instead of the $\tilde X$'s, the $A$'s instead of $\tilde A$'s, etc; these modifications are only important for Lemma~\ref{lem:emergencybound} and not Lemma~\ref{lem:normalbound}.

Note that from the definition~\eqref{eq:Q} we have $Q_{n} = \sqrt{\tilde{M}_{n}^{2} + K \tilde I_{n}^{2}} = \sqrt{\tilde{M}_{n}^{2} + K \tilde I_{n}^{2}}$. We will obtain the bound~\eqref{eq:lemma1} by bounding both $M_{n}$ and $I_{n}$. We first start with obtaining a bound on $M_{n}$

From~\eqref{eq:basic} we have
\begin{equation}
\mathbb E[X_{n+1}^2|\mathcal F_n]=\sigma_A^2 X_n^2+(\mu_AX_n-U_n)^2+\sigma_W^2.
\end{equation}
\noindent In addition, since we are in normal mode we have from~\eqref{eq:I}: that
\[|\mu_A X_n-U_n|\leq \mu_A I_n.\] 
\noindent This implies 
\begin{equation}\label{eq:firstbound} 
\mathbb E[X_{n+1}^2\mathbf{1}_{\tau(n+1)=n+1}|\mathcal F_n]\leq E[X_{n+1}^2|\mathcal F_n]\leq \sigma_A^2M_n^2+\mu_AI_n^2+\sigma_W^2.
\end{equation}
Furthermore, we assume this round lasts exactly one time-step and $\tau(n+1)=n+1$. Hence, we also have that either
\begin{equation}
|M_{n+1}-X_{n+1}|\leq P\delta M_n, \label{eq:abovelimit}
\end{equation}
\noindent or the state $X_{n+1}$ has gotten below the lower limit $M_0$ and we have that $M_{n+1} = M_0$:
\begin{equation}
|X_{n+1}|\leq M_{n+1}=M_0.\label{eq:belowlimit}
\end{equation}
From this we obtain the error estimate
\begin{equation}
\mathbb E[(M_{n+1}^2-X_{n+1}^2)\mathbf{1}_{\tau(n+1)=n+1}\mid \mathcal F_n] \leq P\delta M_n \left(2\mathbb E[|X_{n+1}|\mathbf{1}_{\tau(n+1)=n+1}\mid \mathcal F_n] + P\delta M_n\right)+M_0^2.\label{eq:secondbound}
\end{equation}
This follows because we have that $(a^2 - b^2) = (a+b) (a-b)$, and further we observe that $|X + M| \leq |2X| + |M-X|.$ The last term $M_0^2$ comes from the case where the state hits the lower limit, and follows from~\eqref{eq:belowlimit}.

Now using the definition of $X_{n+1}$ in~\eqref{eq:system}, the choice of $U_{n}$ in~\eqref{eq:u} and the simple estimate \mbox{$\mathbb E[|Z|]\leq \mu_Z+\sigma_Z$} for any real random variable $Z$, we obtain:
\begin{equation}
\mathbb E[|X_{n+1}|\mid \mathcal F_n] \leq \mathbb E[|A_n|]|X_n|+\mathbb E[|W_n|]+|\mu_A M_n| \leq (2\mu_A+\sigma_A)M_n+\sigma_W. \label{eq:thirdbound}
\end{equation}
Adding~\eqref{eq:firstbound} and~\eqref{eq:secondbound}, and then using the bound in~\eqref{eq:thirdbound} for the conditional expectation of $|X_{n+1}|$ we get that:
\begin{equation}
\mathbb E[M_{n+1}^2\mathbf{1}_{\tau(n+1)=n+1} \mid \mathcal F_n]\leq \left(\sigma_A^2+(2\mu_A+\sigma_A)(2P\delta)+P^2\delta^2\right)M_n^2 +2P\delta\sigma_W M_n+\mu_A I_n^2 +\sigma_W^2+ M_0^2. 
\end{equation}
Since $2P\delta\sigma_WM_n\leq P^2\delta^2M_n^2+\sigma_W^2$ we can eliminate the term that is linear in $M_n$, obtaining the slightly cleaner
\begin{equation}\label{eq:mbound}
\mathbb E[M_{n+1}^2\mathbf{1}_{\tau(n+1)=n+1}|\mathcal F_n]\leq \left(\sigma_A^2+(2\mu_A+\sigma_A)(2P\delta)+2P^2\delta^2\right)M_n^2 +\mu_A I_n^2 +2\sigma_W^2+M_0^2.
\end{equation}

Estimating $I_{n+1}^2$ is much simpler; since we know from~\eqref{eq:idef} that $|I_{n+1}|\leq \max(P\delta M_n,M_0)$ whenever $\tau(n)=n$ we obviously have
\begin{equation}\label{eq:ibound}
\mathbb E[I_{n+1}^2\mathbf{1}_{\tau(n)=n}|\mathcal F_n]\leq P^2\delta^2 M_n^2+M_0^2.
\end{equation}

Now, recall we defined $Q_n=\sqrt{\tilde M_n^2+K\tilde I_n^2}.$
Adding~\eqref{eq:mbound} and~\eqref{eq:ibound} we get:
\begin{equation} \label{eq:qbound}
\mathbb E[Q_{n+1}^2\mathbf{1}_{\tau(n+1)=n+1}|\mathcal F_n]\leq \left(\sigma_A^2+(2\mu_A+\sigma_A)(2P\delta)+(2+K)P^2\delta^2\right)M_n^2 +\mu_A I_n^2 +2\sigma_W^2+(K)M_0^2. 
\end{equation}
Now, we choose parameters $K, \delta$ so that we can bound the coefficients of $M_{n}^{2}$ and $I_{n}^{2}$ in~\eqref{eq:qbound}, i.e, so that the following two inequalities hold:
\begin{align}
&\sigma_A^2+(2\mu_A+\sigma_A)(2P\delta)+(2+K)P^2\delta^2 \leq 1-c,\textrm{ and}\\
&\mu_A \leq (1-c)K.
\end{align}
These can be satisfied for any positive $c<1-\sigma_A^2$ if we take $P$ arbitrary, choose $K\geq \frac{\mu_A}{1-c}$, and finally take $\delta=\delta(c,K,P)$ sufficiently small. This gives:
\begin{align*}
\mathbb E[Q_{n+1}^2\mathbf{1}_{\tau(n+1)=n+1}|\mathcal F_n]&\leq (1-c) M_n^2 +(1-c)K I_n^2 +2\sigma_W^2+(K)M_0^2\\
&= (1-c)Q_n^2 + 2\sigma_W^2+(K)M_0^2.
\end{align*}
\end{proof}

\begin{rem}

Making the choice $c<1-\sigma_A^2$ is technically the only place in the proof where we use the necessary assumption $\sigma_A^2<1$. The value for $P$ was essentially irrelevant in the above calculations but picking $P$ appropriately will be important for Lemma~\ref{lem:emergencybound}. In particular, it is important that an arbitrarily large value of $P$ is acceptable if we take $\delta$ very small to compensate. The value $M_0$ did not come into play but will play a minor role in proving Lemma~\ref{lem:emergencybound}. \\

For some intuition on the definition of $Q_n$, note that if $M_n=I_n$, i.e. when we hit the lower limit of $M_{0}$, then we cannot guarantee an averaged squared decrease of $M_n\to M_{n+1}$. However, in this case we expect $I_{n+1}= P\delta I_n$ to hold. In the other extreme case when we know $X_n=M_n$, we expect $M_n\to M_{n+1}$ to result in a decrease, and so picking $\delta$ small enough ensures that the regularization term $KI_n^2$ doesn't hurt us too much. Hence in both cases we expect a squared decrease from $Q_n\to Q_{n+1}$ when parameters are chosen as appropriately, e.g. as above.  

\end{rem}

\begin{proof}[Proof of Lemma~\ref{lem:emergencybound}]

Again, here note that $\tau(n) = n$, and we estimate $\mathbb E[N_{n+1}^2\mathbf{1}_{\tau(n+1)>n+1}]$ under this assumption. Note that if $\tau(n+1)=n+k+1$ (for $k\geq 1$) then we have  

\begin{align}
N_{n+1}^2 &\leq 2^{2k}\left(\tilde M_{n+k+1}^2+\tilde I_{n+k+1}^2\right)\nonumber \\
&\leq 2^{2k+1} \tilde M_{n+k+1}^2 = 2^{2k+1}P^{2k+2} \tilde M_n^2 \nonumber \\
&\leq 2^{2k+1} P^{2k+2} N_n^2. \label{eq:nn}
\end{align}
Hence qualitatively, it will suffice to show $\tau(n+1) - (n+1)$ has very fast decaying tails. This is what we will do.

We have $\tilde X_{n+1}=\tilde A_n\tilde X_n+\tilde W_n-\tilde U_n.$
For later emergency rounds, we have $\tilde U_{n+j}=0$ and so
\[\tilde X_{n+j+1}=\tilde A_{n+j}\tilde X_{n+j}+\tilde W_{n+j}.\]
Hence again taking $\tau(n+1)=n+k+1$, for each $k\geq h\geq 0$ we may write
\[\tilde X_{n+h+1}=(\tilde A_{n+1}\tilde A_{n+2}...\tilde A_{n+h})(\tilde A_n\tilde X_n-\tilde U_n) + \sum_{i=0}^h \left(\tilde W_i\prod_{j=i+1}^{h} \tilde A_{n+j}\right). \]
Since the above equation only holds when $\tau(n+1)\geq n+h+1$, for each $h\geq 0$ we define
\begin{equation}
Z_{n+h+1}:=(\tilde A_{n+1}\tilde A_{n+2}...\tilde A_{n+h})(\tilde A_n\tilde X_n-\tilde U_n) + \sum_{i=0}^h \left(\tilde W_i\prod_{j=i+1}^{h} \tilde A_{n+j}\right). 
\end{equation}
Since $|\tilde U_n|\leq |\mu_A|\tilde M_n$ we have 
\begin{equation}
|Z_{n+h+1}|\leq 2\tilde M_n(|\tilde A_n|+ |\mu_A|)(|\tilde A_{n+1}\tilde A_{n+2}...\tilde A_{n+h}|)+ \sum_{i=0}^h \left(|\tilde W_i|\prod_{j=i+1}^{h} |\tilde A_{n+j}|\right). \label{eq:ybound}
\end{equation}
To control the probability that $\tau(n+1)=n+k+1$, we will estimate the $\alpha$-moments of $Z_{n+h+1}$. The point is that $\tau(n+1)=n+k+1$ implies $|Z_{n+k}|\geq P^k\tilde M_n$ which is an abnormally large value. To do this we need to control the moments for each term. For convenience, define
\begin{equation} \label{eq:malpha}
m_{\alpha}=\max\left(2,\mathbb E\left[(|A|+|\mu_A|)^{\alpha}\right]\right),
\end{equation}
and 
\begin{equation}\label{eq:lalpha}
\ell_{\alpha}=\mathbb E[|W|^{\alpha}].
\end{equation}
Some of the $\tilde A_{n+j}$ terms may have the law of $A$, while others may be almost surely $1$. However, the value $m_{\alpha}$ will serve to estimate both cases uniformly; similarly $W$ might be identically $0$, and this will also be fine. We have the simple estimate for the first term in~\eqref{eq:ybound}, which follows from~\eqref{eq:malpha}, the independence of $\tilde{M}_n$ and all of the $\tilde{A}_k$'s, which are i.i.d..
\begin{equation} \label{eq:est1}
\mathbb E\left[\left|2\tilde M_n(|\tilde A_n|+|\mu_A|)(\tilde A_{n+1}\tilde A_{n+2}...\tilde A_{n+h})\right|^{\alpha} \bigr| ~\mathcal{F}_n\right]\leq \tilde M_n^{\alpha} 2^{\alpha}m_{\alpha}^{h+1}. 
\end{equation}
For the individual terms in the summation in~\eqref{eq:ybound}, we similarly have from~\eqref{eq:malpha},~\eqref{eq:lalpha}:
\[ \mathbb E\left[ \left(|\tilde W_i| \prod_{j=i+1}^h |\tilde A_{n+j}|\right)^{\alpha}   \biggr|~\mathcal{F}_n\right] \leq \ell_{\alpha}m_{\alpha}^{h-i}.\]
Let $B_i = \left(|\tilde W_i| \prod_{j=i+1}^h |\tilde A_{n+j}|\right)$. Then, the above inequality can be written as $||B_i||^{\alpha}_{\alpha} \leq \ell_{\alpha}m_{\alpha}^{h-i},$ where the $L_\alpha$-norm is taken with respect to $\mathbb{E}(\cdot \mid \mathcal{F}_n)$. 
The subadditivity of the $L_\alpha$-norm implies that
\[ ||\sum_i B_i||^{\alpha}_{\alpha} \leq \left(\sum_i||B_i||_{\alpha}\right)^{\alpha}. \]
This gives us the following bound for some constant $C(\alpha)$, where we also use $m_{\alpha}^{h-1} \leq m_{\alpha}^{h}$:
\begin{equation}\label{eq:est2}
\mathbb E\left[ \left(\sum_{i=0}^h|\tilde W_i| \prod_{j=i+1}^h |\tilde A_{n+j}|\right)^{\alpha}   \biggr| ~\mathcal F_n\right] \leq C(\alpha) \ell_{\alpha}m_{\alpha}^h.
\end{equation}
To combine the two estimates~\eqref{eq:est1} and~\eqref{eq:est2} we simply note that 
\[|x+y|^{\alpha}\leq 2^{\alpha}(|x|^{\alpha}+|y|^{\alpha})\]
for all reals $x,y$. Hence we obtain the following, where $C_{1}(\alpha)$ and $C_{2}(\alpha)$ are also constants:
\[\mathbb E\left[|Z_{n+h+1}|^{\alpha} \mid \mathcal{F}_n\right] \leq C_{1}(\alpha) (\tilde M_n^{\alpha}m_{\alpha}^{h+1}+ \ell_{\alpha}m_{\alpha}^h) \leq C_{2}(\alpha)m_{\alpha}^h(\tilde M_n^{\alpha}+1).\]
Because we have $|\tilde M_n|\geq M_0$ for all $n$, we may simply say for a constant $C_{3}(\alpha)$: 
\begin{equation}
\mathbb E\left[|Z_{n+h+1}|^{\alpha} \mid \mathcal{F}_n\right] \leq C_{3}(\alpha)m_{\alpha}^h\tilde M_n^{\alpha}. \label{eq:Ybound}
\end{equation}
Now in the event that $\tau(n)=n+k+1$ we must have 
\[|Z_{n+k}|\geq P^k\tilde M_n.\]
Therefore by the Markov inequality, 
\begin{align}
\mathbb P\left[|Z_{n+k}|\geq P^k\tilde M_n \mid \mathcal{F}_n\right] &\leq P^{-k\alpha}\tilde M_n^{-k\alpha} \mathbb E\left[|Z_{n+k}|^{\alpha}\mid \mathcal{F}_n\right]\nonumber\\
&\leq C_{3}(\alpha)P^{-k\alpha} \tilde M_n^{\alpha-k\alpha} m_{\alpha}^k. \label{eq:starstar}
\end{align}
Combining our work, we have
\begin{align}
\mathbb E[|N_{n+1}^2 {\tau(n+1)>n+1}] &\leq \sum_{k\geq 1} \left(2^{2k+1}P^{2k+2}N_n^2 ~\mathbb{P}\left[|Z_{n+k}|\geq P^k\tilde M_n\right] \right)\label{eq:end1}\\
&\leq  C_{4}(\alpha) N_n^2 \sum_{k\geq 1}  \left(2^{2k-2}P^{2k+2-k\alpha}\tilde M_n^{\tau(1-k)}m_{\alpha}^k  \right), \label{eq:end2}
\end{align}
for a fourth constant $C_{4}(\alpha)$.
The first inequality~\eqref{eq:end1} follows from~\eqref{eq:starstar}, and the second~\eqref{eq:end2} follows from~\eqref{eq:nn}.
This sum is a geometric series with first term $P^{4-\alpha}m_{\alpha}$, and ratio, $4P^{2-\alpha}\tilde M_n^{-\alpha}m_{\alpha}.$
Hence we obtain
\[\mathbb E[N_{n+1}^2 \mathbf{1}_{\tau(n+1)>n+1}] \leq C^{(4)}_{\alpha}N_n^2 \left(\frac{P^{4-\alpha}m_{\alpha}}{1-4P^{2-\alpha}\tilde M_n^{-\alpha}m_{\alpha}}\right).\] 
Since everything so far has been uniform in $P$, we see that for any $\alpha>4$, taking $P$ sufficiently large gives the upper bound 
\[\mathbb E[N_{n+1}^2 \mathbf{1}_{\tau(n+1)>n+1}] \leq \varepsilon N_n^2\]
as desired. (Recall $A,W$ have bounded $\alpha$-moment for some $\alpha>4$.) 

\end{proof}

\section{Conclusion and Future work}
Our paper considered the problem of stabilizing a system that was growing unpredictably using observations over a rate-limited channel. We provide a time-varying strategy that is able to stabilize the system. The  controller takes different actions based on whether the value of the system state is in a predicted interval or not --- we believe such a time-varying strategy is essential for this problem. In future work, we aim to close the gap between this strategy and the converse bound in~\cite{kostina_rate-limited_2016}. Extensions to the vector case are also interesting, since the differential growth rate along different directions must be taken into account.
\section*{Acknowledgments}

We thank Mikl\'os R\'acz and Serdar Y\"uksel for interesting discussions regarding this problem. We also thank the ISIT reviewers for their helpful comments.

%
%
%
%
\bibliographystyle{IEEEtran}
\bibliography{variablemultnoise}

\begin{thebibliography}{10}
\providecommand{\url}[1]{#1}
\csname url@samestyle\endcsname
\providecommand{\newblock}{\relax}
\providecommand{\bibinfo}[2]{#2}
\providecommand{\BIBentrySTDinterwordspacing}{\spaceskip=0pt\relax}
\providecommand{\BIBentryALTinterwordstretchfactor}{4}
\providecommand{\BIBentryALTinterwordspacing}{\spaceskip=\fontdimen2\font plus
\BIBentryALTinterwordstretchfactor\fontdimen3\font minus
  \fontdimen4\font\relax}
\providecommand{\BIBforeignlanguage}[2]{{%
\expandafter\ifx\csname l@#1\endcsname\relax
\typeout{** WARNING: IEEEtran.bst: No hyphenation pattern has been}%
\typeout{** loaded for the language `#1'. Using the pattern for}%
\typeout{** the default language instead.}%
\else
\language=\csname l@#1\endcsname
\fi
#2}}
\providecommand{\BIBdecl}{\relax}
\BIBdecl

\bibitem{wong1997systems}
W.~S. Wong and R.~W. Brockett, ``{Systems with Finite Communication Bandwidth
  Constraints---Part I: State Estimation Problems},'' \emph{IEEE Transactions
  on Automatic Control}, vol.~42, no.~9, pp. 1294--1299, 1997.

\bibitem{tatikonda}
S.~Tatikonda and S.~Mitter, ``{Control under communication constraints},''
  \emph{IEEE Transactions on Automatic Control}, vol.~49, no.~7, pp.
  1056--1068, 2004.

\bibitem{nair2000stabilization}
G.~N. Nair and R.~J. Evans, ``{Stabilization with data-rate-limited feedback:
  tightest attainable bounds},'' \emph{Systems \& Control Letters.}, vol.~41,
  no.~1, pp. 49--56, 2000.

\bibitem{nair2007feedback}
G.~N. Nair, F.~Fagnani, S.~Zampieri, and R.~J. Evans, ``{Feedback control under
  data rate constraints: An overview},'' \emph{Proceedings of the IEEE},
  vol.~95, no.~1, pp. 108--137, 2007.

\bibitem{kostina_rate-limited_2016}
V.~Kostina, Y.~Peres, M.~R{\'a}cz, and G.~Ranade, ``Rate-limited control of
  systems with uncertain gain,'' \emph{Allerton Conference on Communication,
  Control, and Computing}, 2016.

\bibitem{uncertaintyThreshold}
M.~Athans, R.~Ku, and S.~Gershwin, ``The uncertainty threshold principle: Some
  fundamental limitations of optimal decision making under dynamic
  uncertainty,'' \emph{IEEE Transactions on Automatic Control}, vol.~22, no.~3,
  pp. 491--495, 1977.

\bibitem{martinsUncertain}
N.~Martins, M.~Dahleh, and N.~Elia, ``{Feedback Stabilization of Uncertain
  Systems in the Presence of a Direct Link},'' \emph{IEEE Trans. Autom.
  Control}, vol.~51, no.~3, pp. 438--447, 2006.

\bibitem{okano2014arxiv}
K.~Okano and H.~Ishii, ``Minimum data rate for stabilization of linear systems
  with parametric uncertainties,'' 2014, arXiv preprint.

\bibitem{victoriacontrol}
V.~Kostina and B.~Hassibi, ``Rate-cost tradeoffs in control,'' 2016, preprint.

\bibitem{phat2004robust}
V.~N. Phat, J.~Jiang, A.~V. Savkin, and I.~R. Petersen, ``Robust stabilization
  of linear uncertain discrete-time systems via a limited capacity
  communication channel,'' \emph{Systems \& Control Letters}, vol.~53, no.~5,
  pp. 347--360, 2004.

\bibitem{okano2012data}
K.~Okano and H.~Ishii, ``Data rate limitations for stabilization of uncertain
  systems,'' in \emph{51st Conference on Decision and Control (CDC)}.\hskip 1em
  plus 0.5em minus 0.4em\relax IEEE, 2012, pp. 3286--3291.

\bibitem{okano2012datalossy}
------, ``{Data Rate Limitations for Stabilization of Uncertain Systems over
  Lossy Channels},'' in \emph{American Control Conference (ACC)}.\hskip 1em
  plus 0.5em minus 0.4em\relax IEEE, 2012, pp. 1260--1265.

\bibitem{controlcapacity}
G.~Ranade and A.~Sahai, ``{Control Capacity},'' in \emph{International
  Symposium on Information Theory (ISIT)}.\hskip 1em plus 0.5em minus
  0.4em\relax IEEE, 2015.

\bibitem{gireejaAllerton}
------, ``Non-coherence in estimation and control,'' in \emph{51st Annual
  Allerton Conf. on Comm., Control, and Comp.}, 2013.

\bibitem{tiger}
J.~Ding, Y.~Peres, and G.~Ranade, ``A tiger by the tail: when multiplicative
  noise stymies control,'' in \emph{International Symposium on Information
  Theory (ISIT)}.\hskip 1em plus 0.5em minus 0.4em\relax IEEE, 2016.

\bibitem{fu2010quantized}
M.~Fu and L.~Xie, ``Quantized feedback control for linear uncertain systems,''
  \emph{International Journal of Robust and Nonlinear Control}, vol.~20, no.~8,
  pp. 843--857, 2010.

\bibitem{zhou1998essentials}
K.~Zhou and J.~C. Doyle, \emph{Essentials of robust control}, 1998.

\bibitem{nairStabilization}
G.~Nair and R.~Evans, ``{Stabilizability of Stochastic Linear Systems with
  Finite Feedback Data Rates},'' \emph{SIAM Journal on Control and
  Optimization.}, vol.~43, no.~2, pp. 413--436, 2004.

\bibitem{brockett2000quantized}
R.~W. Brockett and D.~Liberzon, ``Quantized feedback stabilization of linear
  systems,'' \emph{IEEE transactions on Automatic Control}, vol.~45, no.~7, pp.
  1279--1289, 2000.

\bibitem{yuksel2010fixedrate}
S.~Y\"uksel, ``Stochastic stabilization of noisy linear systems with fixed-rate
  limited feedback,'' \emph{IEEE Transactions on Automatic Control}, vol.~55,
  no.~12, pp. 2847--2853, 2010.

\bibitem{kostina2018exact}
V.~Kostina, Y.~Peres, G.~Ranade, and M.~Sellke, ``Exact minimum number of bits
  to stabilize a linear system,'' 2018.

\end{thebibliography}
\end{document}